\def\year{2021}\relax
\documentclass[letterpaper]{article} % DO NOT CHANGE THIS
\usepackage{aaai21}  % DO NOT CHANGE THIS
\usepackage{times}  % DO NOT CHANGE THIS
\usepackage{helvet} % DO NOT CHANGE THIS
\usepackage{courier}  % DO NOT CHANGE THIS
\usepackage[hyphens]{url}  % DO NOT CHANGE THIS
\usepackage{graphicx} % DO NOT CHANGE THIS
\usepackage{natbib}  % DO NOT CHANGE THIS AND DO NOT ADD ANY OPTIONS TO IT
\usepackage{caption} % DO NOT CHANGE THIS AND DO NOT ADD ANY OPTIONS TO IT
\usepackage{amsmath,amssymb,amsthm}
\usepackage{tabularx}
\usepackage[table]{xcolor}
\usepackage{multirow}
\usepackage{booktabs}
\usepackage{multicol}

\newtheorem{theorem}{Theorem}
\newtheorem{corollary}{Corollary}

\newcommand{\np}{{\textsf{NP}}}

\newcommand{\nph}{{\np}-{\textsf{hard}}}
\newcommand{\nphns}{{\np}-{\textsf{hardness}}}

\newcommand{\poly}{{\textsf{P}}}
\newcommand{\wah}{{\textsf{W[1]-hard}}}

\newcommand{\edge}[2]{(#1, #2)}
\newcommand{\fpt}{{\sf{FPT}}}

\newcommand{\no}{No}

\newcommand{\noins}{No-instance}
\newcommand{\bigo}{O}
\newcommand{\bigos}{O^*}

\newcommand{\abs}[1]{\left\vert #1 \right\vert}

\newcommand{\ver}[1]{{\sf{V}}(#1)} %% vertex set of a graph #1
\newcommand{\onlyfull}[1]{}
\newcommand{\setmid}{:}

\newcommand{\cclique}[1]{\mathcal{K}(G)}

\newcommand{\prob}[1]{{\textsc{#1}}}

\newcommand{\EP}[3]{
\begin{center}
{\small
\begin{tabularx}{0.98\columnwidth}{ll}
\toprule
\multicolumn{2}{c}{\textsc{#1}} \\
\midrule
{\bf Given:}   & \parbox[t]{0.75\columnwidth}{#2\vspace*{0mm}}  \\
{\bf Question:}& \parbox[t]{0.75\columnwidth}{#3\vspace*{0mm}} \\
\bottomrule
\end{tabularx}
}
\end{center}
}
\title{A Refined Study of the Complexity of Binary Networked Public Goods Games}
\author {
Yongjie Yang\textsuperscript{\rm 1}
        Jianxin Wang\textsuperscript{\rm 2}\\
}
\affiliations {
    \textsuperscript{\rm 1} Chair of Economic Theory, Saarland University, Saarbr\"{u}cken, Germany \\
    \textsuperscript{\rm 2} School of Computer and Engineering, Central South University, Changsha, China \\
    yyongjiecs@gmail.com, jxwang@csu.edu.cn
    }
\begin{document}

\maketitle

\begin{abstract}
We study the complexity of several combinatorial problems in the model of binary networked public goods games. In this game, players are represented by vertices in a network, and the action of each player can be either investing or not investing in a public good. The payoff of each player is determined by the number of neighbors who invest and her own action. We study the complexity of computing action profiles that are Nash equilibrium, and those that provide the maximum utilitarian or egalitarian social welfare. We show that these problems are generally {\nph} but become polynomial-time solvable when the given network is restricted to some special domains, including networks with a constant bounded treewidth, and those whose critical clique graphs are forests.
\end{abstract}

\hyphenation{BNPG}
\hyphenation{BNPGs}
\hyphenation{PSNE}
\hyphenation{PSNEC}

\section{Introduction}
Binary networked public good games (BNPGs) model the scenario where players reside in a network and decide whether they invest in a public good. The value of the public good is determined by the total amount of investment and is shared by players in some specific way. Particularly, the payoff of each player is determined by the number of her neighbors who invest and her own action. BNPGs are relevant to several real-world applications. One example could be vaccination, where parents decide whether to vaccinate their children, and the public good here is herd immunity.

In a recent paper, Yu~et~al.~\shortcite{DBLP:conf/aaai/YuZBV20} explored the complexity of determining the existence of pure-strategy Nash equilibria (PSNE) in BNPGs. In particular, they showed that this problem is {\nph}. On the positive side, they derived polynomial-time algorithms for the case where the given network is a clique or a tree. Following their work, we embark on a more refined complexity study of BNPGs. Our main contributions are as follows.
\begin{itemize}
\item We fix a flaw in an {\nphns} proof in~\cite{DBLP:conf/aaai/YuZBV20}.

\item Besides PSNE (action) profiles, we consider also  profiles with the maximum utilitarian or egalitarian social welfare which are not the main focus of~\cite{DBLP:conf/aaai/YuZBV20}.

\item We show that the problem of computing profiles with the maximum utilitarian/egalitarian social welfare is {\nph}, and this holds even when the given network is bipartite and has a constant diameter.

\item For the problems studied, we derive polynomial-time algorithms restricted to networks whose critical clique graphs are forests, or whose treewidths are bounded by a constant. Note that the former class of graphs contains both trees and cliques, and the latter class contains trees. Therefore, our results widely extend the {\poly} results studied in~\cite{DBLP:conf/aaai/YuZBV20}.
\end{itemize}

{\bf{Related Works.}} BNPGs are a special case of graphical games proposed by Kearns, Littman, and Singh~\shortcite{DBLP:conf/uai/KearnsLS01}, where the complexity of computing Nash equilibria has been investigated in the literature~\cite{DBLP:conf/uai/KearnsLS01,DBLP:conf/sigecom/ElkindGG06}. Recall that graphical games are proposed as a succinct representation of $n$-player $2$-action games whose action space is $2^n$. In graphical games, players are mapped to vertices in a network and the payoff of each player is entirely determined by her own action and those of her neighbors. BNPGs are specified so that, in addition to one's own action, only the number of neighbors who invest matters for the payoff of this player. In other words, in BNPGs, every player treats her neighbors equally. An intuitive generalization of BNPGs where each player may take more than two actions have been considered by Bramoull\'{e} and Kranton~\shortcite{DBLP:journals/jet/BramoulleK07}. Additionally, BNPGs are also related to supermodular network games~\cite{ManshadiJohariAllerton2009} and best-shot games~\cite{HarrisonH1989,Carpenter2002,DBLP:journals/ai/LevitKGM18}. For more detailed discussions or comparisons we refer to~\cite{DBLP:conf/aaai/YuZBV20}.

\section{Preliminaries}
We assume that the reader is familiar with basic notions in graph theory and computational complexity, and, if not, we refer to~\cite{Douglas2000,DBLP:journals/interfaces/Tovey02}.% for a consult.  %Below we provide some other basic notions that are relevant to our study.

Let $G=(V, E)$ be an undirected graph. The vertex set of~$G$ is also denoted by~$\ver{G}$. The set of (open) neighbors of a vertex in $v\in V$ in~$G$ is denoted by $N_G(v)=\{v'\in V \setmid \edge{v}{v'}\in E\}$. The set of closed neighbors of~$v$ is $N_G[v]=N_G(v)\cup \{v\}$.
%The degree of~$v$, denoted by~$\deg{v}$, is the cardinality of~$N_G(v)$.
We define~$n_G(v)$ (resp.\ $n_G[v]$) as the cardinality of~$N_G(v)$ (resp.\ $N_G[v]$). The notion~$n_G(v)$ is often called the degree of~$v$ in~$G$ in the literature.
Given a subset $V'\subseteq V$, we use $n_G(v, V')=\abs{N_G(v)\cap V'}$ (resp.\ $n_{G}[v, V']=\abs{N_G[v]\cap V'}$) to denote the number of open neighbors (resp.\ closed neighbors) of~$v$ contained in the set~$V'$. The subgraph induced by~$V'\subseteq V$ is denoted by~$G[V']$, and $G-V'$ is the subgraph induced by $V\setminus V'$.

A BNPG~$\mathcal{G}$ is a $4$-tuple $(V, E, g_V, c)$. In the notion,~$V$ is a set of players, and~$(V, E)$ is an undirected graph with~$V$ being the vertex set. In addition,~$g_V$ is a set of  functions, one for each player. In particular, for every player $v\in V$ there is one function $g_v: \mathbb{N}_0 \rightarrow \mathbb{R}_{\geq 0}$ in~$g_V$ such that~$g_v(i)$, where~$i$ is a nonnegative integer not greater than~$n_G[v]$, measures the external benefits of the player~$v$ when exactly~$i$ of her {\bf{closed neighbors}} invest. Thus, the function~$g_v$ is called the externality function of~$v$. Finally, $c: V\rightarrow \mathbb{R}_{\geq 0}$ is a cost function where~$c(v)$ is the investment cost of player~$v\in V$.

For a function $f: A\rightarrow B$ and $A'\subseteq A$, we use~$f|_{A'}$ to denote the function~$f$ restricted to~$A'$, i.e., it holds that $f_{A'}: A' \rightarrow B$ such that $f_{A'}(v)=f(v)$ for all $v\in A'$. Sometimes we also write $f|_{\neg (A\setminus A')}$ for~$f|_{A'}$. For $V'\subseteq V$, a subgame of~$\mathcal{G}=(V, E, g_V, c)$ induced by~$V'$, denoted by $\mathcal{G}|_{V'}$, is the game $(V', E', g_{V'}, c|_{V'})$ such that $E'=\{\edge{u}{u'}\in E \setmid u, u'\in V'\}$ and $g_{V'}=\{g_{v} \setmid v\in V'\}$.

An (action) profile of a BNPG is represented by a subset consisting of the players who invest. Given a profile $\mathbf{s}\subseteq V$, the payoff of every player $v\in V$ is defined by
\[\mu(v, {\bf{s}}, \mathcal{G})=g_v(n_G[v, {\bf{s}}])-{\mathbf{1}}_{\bf{s}}(v) \cdot c(v),\]
where ${\mathbf{1}}_{\bf{s}}(\cdot)$ is the indicator function of ${\bf{s}}$, i.e., ${\bf{1}}_{\bf{s}}(v)$ is~$1$ if $v\in \bf{s}$ and is~$0$ otherwise.
If it is clear which BNPG is discussed, we drop\onlyfull{ the third parameter}~$\mathcal{G}$ from~$\mu(v, {\bf{s}}, \mathcal{G})$ for brevity.
\smallskip

{\bf{Nash equilibria.}}
In the setting of noncooperative game, it is assumed that players are all self-interested. In this case, profiles that are stable in some sense is of particular importance.
%The concept of Nash equilibrium  if no one has an incentive to alter his action. %Its formal definition is given below.
%
For a profile~${\bf{s}}$ and a player $v\in V$, let ${\bf{s}}(\neg v)$ be the profile obtained from~${\bf{s}}$ by altering the action of~$v$, i.e., $v \in {\bf{s}}(\neg v)$ if and only if $v\not\in {\bf{s}}$, and for every $v'\in V\setminus \{v\}$ it holds that $v'\in {\bf{s}}(\neg v)$ if and only if $v'\in {\bf{s}}$.
We say that a player $v\in V$ has an incentive to deviate under~${\bf{s}}$ if it holds that $\mu({\bf{s}}, v, \mathcal{G})< \mu({\bf{s}}(\neg v), v, \mathcal{G})$.
A profile~${\bf{s}}$ is a PSNE if none of the players has an incentive to deviate under~${\bf{s}}$.
\smallskip

{\bf{Social welfare.}}
In some other settings, there is an organizer, a leader, or an authority (e.g., in the game of vaccination the authority might be the government) who is able to coordinate the actions of the players in order to yield the maximum possible welfare for the whole society. We consider two types of social welfare as follows. Let $\mathcal{G}=(V, E, g_V, c)$ be a BNPG and~${\bf{s}}$ a profile of~$\mathcal{G}$.
\begin{description}
\item[Utilitarian social welfare (USW)] The USW of $\bf{s}$, denoted by ${\textsf{USW}}({\bf{s}}, \mathcal{G})$, is the sum of all utilities of all players:
\[{\textsf{USW}}({\bf{s}}, \mathcal{G})=\sum_{v\in V}\mu({\bf{s}}, v, \mathcal{G}).\]

\item[Egalitarian social welfare (ESW)] Under profiles with the maximum USW, it can be the case where some player receives extremely large utility while others obtain only little utility, which is considered unfair in some circumstances. Unlike USW, ESW cares about the player with the least payoff. Particularly, the ESW of~${\bf{s}}$ is
    \[{\textsf{ESW}({\bf{s}}, \mathcal{G})}=\min_{v\in V}\mu({\bf{s}}, v, \mathcal{G}).\]
\end{description}

{\bf{Problem formulation.}}
We study the following problems which have the same input $\mathcal{G}=(V, E, g_V, c)$. The notations of the problems and their tasks are as follows.
\begin{description}
\item \prob{PSNE Computation (PSNEC)}: Compute a PSNE profile of~$\mathcal{G}$ if there are any, and output ``{\no}'' otherwise.

\item \prob{USW/ESW Computation (USWC/ESWC)}: Compute a profile of~$\mathcal{G}$ with the maximum USW/ESW.
\end{description}

{\bf{Some hard problems.}} Our hardness results are established based on the following problems.

For a positive integer~$d$, a $d$-regular graph is a graph where every vertex has degree exactly~$d$.
\EP
{$3$-Regular Induced Subgraph ($3$-RIS)}
{A graph~$G$.}
{Does~$G$ admit a $3$-regular induced subgraph?}

The {\prob{$3$-RIS}} problem is {\nph}~\cite{DBLP:journals/tcs/BroersmaGP13,DBLP:journals/dam/CheahC90}.

A clique of a graph is a subset of vertices such that between every pair of vertices there is an edge.

\EP
{$\kappa$-Clique}
{A graph~$G$ and a positive integer~$\kappa$.}
{Is there a clique of size~$\kappa$ in~$G$?}

{\prob{$\kappa$-Clique}} is a well-known {\nph} problem~\cite{DBLP:conf/coco/Karp72}.
%The problem below is also {\nph}~\cite{garey}.
A vertex~$v$ dominates another vertex~$v'$ in a graph~$G$ if there is an edge between~$v$ in~$v'$ in~$G$. In addition, a subset~$A$ of vertices dominate another subset~$B$ of vertices if every vertex in~$B$ is dominated by at least one vertex in~$A$.
\EP
{Red-Blue Dominating Set (RBDS)}
{A bipartite graph $G=(B\uplus R, E)$ and a positive integer~$\kappa$.}
{Is there a subset $B'\subseteq B$ such that~$\abs{B'}\leq \kappa$ and~$B'$ dominates~$R$?}

{The {\prob{RBDS}} problem is {\nph}~\cite{garey}.}

We assume that in instances of {\prob{$\kappa$-Clique}} and {\prob{RBDS}}, the graph~$G$ does not contain any isolated vertices (vertices of degree one).

\section{BNPG Solutions in General}
In this section, we settle the complexity of the {\prob{PSNEC/USWC/ESWC}} problems.
We first point out a flaw in the proof of Theorem~10 in~\cite{DBLP:conf/aaai/YuZBV20}, where an {\nphns} reduction for {\prob{PSNEC}} restricted to fully-homogeneous BNPGs is established. \footnote{Due to space limitation they omitted the proof in~\cite{DBLP:conf/aaai/YuZBV20} but the proof is included in a full version posted online (\url{https://arxiv.org/pdf/1911.05788.pdf}). [19.11.2020]} Recall that a BNPG is fully-homogenous if both the externality functions and the investment costs of all players are the same. 
%Their proof is flawed in two aspects. 
A profile is trivial if either all players invest or none of the players invests.
%, i.e., it is either the empty set or the set of all players.
Checking whether a trivial profile is PSNE is clearly easy. 
%Based on this observation, 
Yu~et.~al.\ provide a reduction to show that determining whether a fully-homogenous BNPG admits a nontrivial PSNE is {\nph}. 
%, and from this they conclude that the {\prob{PSNEC}} problem restricted to fully-homogenous BNPGs is {\nph}.
%However, it is not clear why the {\nphns} of the former implies that of the latter. 
%Besides this, the reduction for determining the existence of a nontrivial PSNE in~\cite{DBLP:conf/aaai/YuZBV20} is itself flawed too. 
However, this reduction is flawed. Let us briefly reiterate the reduction, which is from the {\prob{$\kappa$-Clique}} problem. The instance of BNPG is obtained from an instance $(G=(V,E), \kappa)$ of {\prob{$\kappa$-Clique}} by first considering vertices in~$V$ as players and then adding a large set~$T$ of~$M$ players who form a clique and are adjacent to all players in~$V$. The externality functions and costs of players are set so that every player is indifferent between investing and not investing if exactly $\kappa-1+M$ of her open neighbors invest, otherwise, the player prefers to not investing. If there is a clique of size~$\kappa$, then there is a PSNE. However, the other direction does not work.
%Particularly, it can be the case that a Nash equilibria nontrivial action profiles exists but there are no cliques of size $\kappa$.
In fact, a nontrivial PSNE needs only the existence of a $(\kappa-1)$-regular induced subgraph of~$G$ (in this case, all players in~$M$ and the players in the regular subgraph invest). A more concrete counterexample is demonstrated in Figure~\ref{fig-reduction-flawed}. Our amendment is as follows.

\begin{figure}
\centering
{
\includegraphics[width=0.3\textwidth]{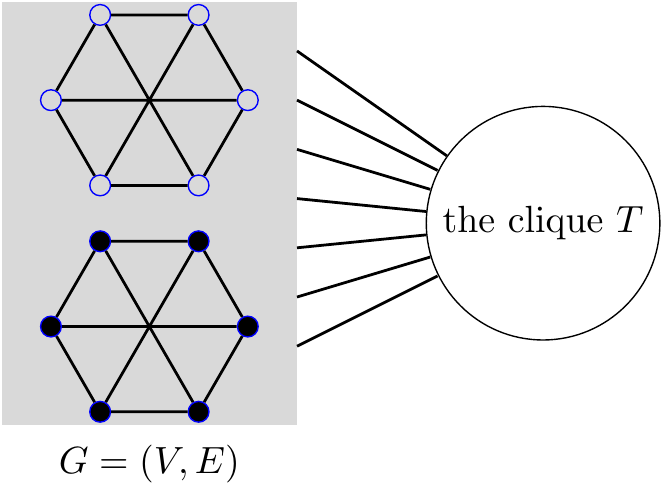}
}
\caption{A counterexample to the reduction in the proof of Theorem~10 in~\cite{DBLP:conf/aaai/YuZBV20}. The graph part~$G$ in the {\prob{$\kappa$-Clique}} instance  is a $3$-regular graph  without cliques of size~$3$, as shown in the gray area. Additionally, $\kappa=3$. The clique~$T$ contains more than~$12$ vertices. The profile consisting of the black-filled players in~$G$ and all those in~$T$ is a nontrivial PSNE.}
\label{fig-reduction-flawed}
\end{figure}

%The proof of Theorem 10 in ``Computing Equilibria in Binary Networked Public Goods Games'' is flawed. We provide a new proof as follows.
%For a network~$G$, a partition $(V_1, V_2)$ of~$V(G)$ represents a status such that all players in~$V_1$ invest and one of players in~$V_2$ invests.

\begin{theorem}
Determining if a fully-homogeneous BNPG admits a nontrivial PSNE is {\nph}.
\end{theorem}

\begin{proof}
We prove the theorem via a reduction from the {\prob{$3$-RIS}} problem. Let $(G, \kappa)$ be a {\prob{$3$}-RIS} instance.
The network is exactly~$G$. We set the externality functions and the costs of the players so that everyone is indifferent between investing and not investing if  exactly three of her open neighbors invest, and prefers to not investing otherwise. This can be achieved by setting, for every $v\in V$,  $c(v)=2$ and
\begin{equation*}
g_v(x)=
\begin{cases}
x & x\leq 3\\
x+1 & 4\leq x\leq n_G[v]\\
\end{cases}
\end{equation*}
The correctness is easy to see. If the graph~$G$ admits a $3$-regular subgraph induced by a subset~$H\subseteq V$, then~$H$ is a nontrivial PSNE. On the other hand, if there is a nontrivial PSNE $\emptyset\neq H\subseteq V$, due to the above construction, every vertex in~$H$ must have exactly~$3$ open neighbors in~$H$, and so~$G[H]$ is $3$-regular.
\end{proof}

Now we start the exploration on profiles providing the maximum social welfare. We show that the corresponding problems are all hard to solve, and this holds even when the given network is a bipartite graph with a constant bounded diameter. Recall that the diameter of a graph is the maximum possible distance between vertices, where the distance between two vertices is defined as the length of a shortest path between them.

\begin{theorem}
\label{thm-uswo-np-hard}
{\prob{USWC}} is {\nph}. This holds even when the given network is bipartite and has diameter at most~$4$.
\end{theorem}

\begin{proof}
We prove the theorem via a reduction from the {\prob{$\kappa$-Clique}} problem to the decision version of {\prob{USWC}} which consists in determining whether there is a profile of utilitarian social welfare at least a threshold value.

Let $(H, \kappa)$ be a {\prob{$\kappa$-Clique}} instance, where $H=(U, E)$ is a graph, $n=\abs{U}$, and $m=\abs{E}$. We assume that~$\kappa$ is considerably smaller than~$m$, say $(k+1)^{10}<m$. As {\prob{$\kappa$-Clique}} is {\wah} with respect to~$\kappa$, this assumption does not change the hardness of the problem. We construct the following instance. For each vertex $u\in U$, we create a player denoted by~$v(u)$. The externality function of~$v(u)$ is defined so that $g_{v(u)}(0)=1$ and $g_{v(u)}(x)=0$ for positive integers~$x\leq n_H[u]$, and the investment cost of~$v(u)$ is $c(v(u))=1$. Let $V(U)=\{v(u) \setmid u\in U\}$ be the set of the vertex-players. In addition, for each edge $e\in E$, where $e=\edge{u}{u'}$, we create a player~$v(e)$. We define $g_{v(e)}(0)=1$ and $g_{v(e)}(x)=0$ for all~$x\in [4]$. Moreover, $c(v(e))=0$. Let $V(E)=\{v(e) \setmid e\in E\}$ be the set of edge-players. Finally, we create a player~$v^*$ such that $g_{v^*}(\frac{\kappa \cdot (\kappa-1)}{2})=m$ and $g_{v^*}(x)=0$ for all other possible values of~$x$. The investment cost of~$v^*$ is $c(v^*)=m$. The player network is a bipartite graph with the vertex partition $(V(U)\cup \{v^*\}, V(E))$. In particular, the edges in the network are as follows: for every edge $e=\edge{u}{u'}\in E$, the player~$v(e)$ is adjacent to exactly~$v(u)$,~$v(u')$, and~$v^*$, and thus has degree~$3$ in the network. It is clear that the network has exactly~$3m$ edges and has diameter at most~$4$.

The construction clearly can be done in polynomial time.
We claim that there is a clique of size~$\kappa$ in the graph~$H$ if and only if there is a profile of USW at least
\[q=(n-\kappa)+\left(m-\frac{\kappa\cdot (\kappa-1)}{2}\right)+m.\]

$(\Rightarrow)$ Assume that there is a clique $K\subseteq U$ of size~$\kappa$ in the graph~$H$. Let $E(K)=\{\edge{u}{u'}\in E \setmid u, u'\in K\}$ be the set of edges in the subgraph induced by~$K$. Clearly,~$E(K)$ consists of exactly $\frac{\kappa\cdot (\kappa-1)}{2}$ edges. Let ${\bf{s}}=\{v(e) \setmid e\in E(K)\}$ be the set of the $\frac{\kappa \cdot (\kappa-1)}{2}$ players corresponding to the edges in~$E(K)$. We claim that~${\bf{s}}$ has USW at least~$q$. From the above construction, the utility of each player~$v(u)$, where $u\in U$, is~$1$ if $u\not\in K$ and is~$0$ otherwise. Hence, the total utility of the vertex-players is exactly $n-\kappa$. In addition, the utility of a player~$v(e)$, $e\in E$, is~$1$ if $e\not\in E(K)$ and is~$0$ otherwise. Therefore, the total utility of edge-players is $m-\frac{\kappa\cdot (\kappa-1)}{2}$. Finally, the utility of the player~$v^*$ is exactly~$m$. The sum of the above utility is exactly~$q$.

$(\Leftarrow)$ Assume that there is a profile with USW at least~$q$. Due to the large investment cost of the player~$v^*$, to maximize the USW,~$v^*$ must not invest and, moreover, by the externality functions, exactly $\frac{\kappa\cdot (\kappa-1)}{2}$ edge-players must invest. Additionally, by the specific setting of the externality functions and the costs of the vertex-players, none of the vertex-players invests in any profile with the maximum USW. It follows that in a profile with the maximum USW, exactly $\frac{\kappa\cdot (\kappa-1)}{2}$ edge-players invest. Then, due to the setting of the externality functions, the smaller the number of vertex-players dominated by the investing edge-players, the larger is the USW. It is then easy to check that a profile achieves a USW at least~$q$ if and only if at most~$\kappa$ vertex-players are dominated by the investing edge-players. This implies that the edges whose corresponding players invest in a profile with USW at least~$q$ induce a clique in~$H$.
\end{proof}

\onlyfull{
If we seek a strategy with maximum social welfare under the restriction that there are at least~$k$ players invest, we have a W[1]-hard result with respect to the number of investors.

\begin{theorem}
Determining whether there is a strategy of size exactly/at least $n-k$ of utilitarian social welfare at least~$q$ is W[2]-hard with respect to~$k$. This holds even for fully-homogeneous case.
\end{theorem}

\begin{proof}
Reduction from dominating set restricted to~$\ell$-regular graphs. Each vertex is a player. Investing does not cost anything for all players. Moreover, for all players $v\in V$, it holds that $g_v(\ell+1)=0$ and $g_v(x)=1$ for all nonnegative integer $x\leq \ell$, and set $q=n$.

It can be also adapted for egalitarian social welfare by resetting $q=1$.
\end{proof}
}

For the computation of profiles with the maximum egalitarian social welfare, we have the same result.

\begin{theorem}
\label{thm-eswo-np-hard}
{\prob{ESWC}} is {\nph}. This holds even when the given network is bipartite with diameter at most~$4$ and all players have the same investment cost.
\end{theorem}

%\onlyfull
{
\begin{proof}
Our proof is based on a reduction from the {\prob{RBDS}} problem. Let $(G, \kappa)$ be an {\prob{RBDS}} instance, where~$G$ is a bipartite graph with the vertex partition $(R\uplus B)$. We construct an instance of the decision version of {\prob{ESWC}}, which takes as input a BNPG and a number~$q$, and determines if the give BNPG admits a profile of ESW at least~$q$.
For each vertex $v\in B\uplus R$, we create one player denoted still by the same symbol~$v$ for simplicity.
In addition, we create a player~$v^*$. The network of the players is obtained from~$G$ by first adding~$v^*$ and then creating edges between~$v^*$ and all blue-players in~$B$, which is clearly a bipartite graph with the vertex partition $(R\cup \{v^*\}, B)$. Furthermore, the diameter of the network is at most~$4$.
The externality and the cost functions are defined as follows.
\begin{itemize}
\item For every red-player~$v\in R$, we define $g_{v}(0)=0$ and $g_{v}(x)=1$ for all other possible integers~$x$.
\item For every blue-player~$v\in B$, we define $g_{v}(0)=1$, $g_{v}(1)=2$, and $g_{v}(x)=0$ for all other possible integers $x\geq 2$.
\item For the player~$v^*$, we define $g_{v^*}(x)=1$ for every nonnegative integer $x\leq \kappa$ and $g_{v^*}(x)=0$ for all other possible integers $x> \kappa$.
\item All players have the same investment cost~$1$, i.e., $c(v)=1$ for every player~$v$ constructed above.
\end{itemize}
The reduction is completed by setting $q=1$.
The above instance can be constructed in polynomial time.
It remains to show the correctness of the reduction.

$(\Rightarrow)$ Assume that there is a subset $B'\subseteq B$  such that $\abs{B'}\leq \kappa$ and every red-player has at least one neighbor in~$B'$. One can check that profile~$B'$ has ESW at least one. Particularly, as~$B'$ is the set of the investing neighbors of~$v^*$, due to the definitions of the externality and cost functions given above, the utility of the player~$v^*$ is $g_{v^*}(\abs{B'})=1$. Let~$v$ be a player other than~$v^*$. If $v\in R$, then as~$v$ has at least one neighbor in~$B'$, the utility of~$v$ is exactly one. If $v\in B'$, then the utility of~$v$ is $g_v(2)-c(v)=1$. Finally, if $v\in B\setminus B'$, the utility of~$v$ is $g_v(0)=1$.

$(\Leftarrow)$ Assume that there is a profile~${\bf{s}}$ where every player obtains utility at least one. Observe first that none of $R\cup \{v^*\}$ can be contained in~${\bf{s}}$, since for every player in $R\cup \{v^*\}$, the investment cost is exactly one and the externality benefit is at most one. It follows that ${\bf{s}}\subseteq B$. Then, it must hold that $\abs{\bf{s}}\leq \kappa$, since otherwise the utility of the player~$v^*$ can be at most $g_{v^*}(\abs{\bf{s}})=0$. Finally, as every player $v\in R$ obtains utility at least one under~$\bf{s}$, at least one of~$v$'s neighbors must be contained in ${\bf{s}}$. This implies that~$\bf{s}$ dominates~$R$.
\end{proof}
}

A close look at the reduction in the proof of Theorem~\ref{thm-eswo-np-hard} reveals that if the {\prob{RBDS}} instance is a {\noins}, the best achievable ESW in the constructed BNPG is zero. The following corollary follows.

\begin{corollary}
{\prob{ESWC}} is not polynomial-time approximable within factor~$\beta(p)$ unless $\poly= \np$, where~$p$ is the input size and~$\beta$ can be any computable function in~$p$. Moreover, this holds even when the given network is bipartite with diameter at most~$4$ and all players have the same investment cost.
\end{corollary}

\section{Games with Critical Clique Forests}
In the previous section, we showed that all problems studied in the paper are {\nph}. This motivates us to study the problems when the give network is subject to some restrictions. Yu~et.~al.~\shortcite{DBLP:conf/aaai/YuZBV20} considered the cases where the given networks are cliques or trees, and showed  separately that {\prob{PSNEC}} in both cases becomes polynomial-time solvable. To significantly extend their results, we derive a polynomial-time algorithm which applies to a much larger class of networks containing both cliques and trees. Generally speaking, we consider the networks whose vertices can be divided into disjoint cliques and, moreover, contracting these cliques results in a forest. For formal expositions, we need the following notions.

A {\it{critical clique}} in a graph $G=(V, E)$ is a clique $K\subseteq V$ whose members share exactly the same neighbors and is maximal under this property, i.e., for every $v\in V\setminus K$ either~$v$ is adjacent to all vertices in the clique~$K$ or is adjacent to none of them, and there does not exist any other clique~$K'$ satisfying the same condition and $K\subset K'$. The concept of critical cliques was coined by Lin, Jiang, and Kearney~\shortcite{DBLP:conf/isaac/LinJK00}, and since then has been used to derive many efficient algorithms (see, e.g.,~\cite{DBLP:journals/tcs/Guo09,DBLP:journals/algorithmica/DomGHN06,DBLP:journals/dam/DomGHN08}). It is well-known any two different critical cliques do not intersect. In addition, for two critical cliques, either they are completely adjacent (i.e., there is an edge between every two vertices from the two cliques respectively), or they are not adjacent at all (i.e., there are no edges between these two cliques). For brevity, when we say two critical cliques are adjacent we mean that they are completely adjacent.
For a graph~$G$, its critical clique graph, denoted by~$\cclique{G}$, is the graph whose vertices are critical cliques of~$G$ and, moreover, there is an edge between two vertices if and only if the corresponding critical cliques are adjacent in~$G$. See Figure~\ref{fig-critical-clique-graphs} for an illustration. Every graph has a unique critical clique graph and, importantly, it can be constructed in polynomial time~\cite{DBLP:conf/isaac/LinJK00,DBLP:journals/algorithmica/DomGHN06}.

\begin{figure}
\centering
{\includegraphics[width=0.3\textwidth]{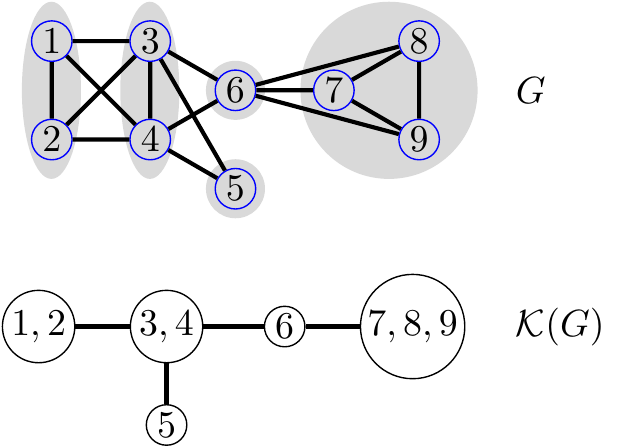}}
\caption{A graph~$G$ and its critical clique graph $\cclique{G}$.}
\label{fig-critical-clique-graphs}
\end{figure}

%\begin{lemma}[\cite{DBLP:conf/isaac/LinJK00,DBLP:journals/algorithmica/DomGHN06}]
%\label{lem-critical-clique-graph}
%Given a graph, its critical clique graph can be constructed in~$\bigo(n+m)$ time where~$n$ and~$m$ are respectively the number of vertices and the number of edges.
%\end{lemma}

We are ready to show the first main result in this section.

\begin{theorem}
\label{thm-psnec-poly-critical-clique-graphs}
{\prob{PSNEC}} is polynomial-time solvable when the critical clique graph of the given network is a forest.
\end{theorem}

\begin{proof}
To prove the theorem, we derive a dynamic programming algorithm to solve {\prob{PSNEC}} in the case where the given network has a critical clique graph being a forest.
Let $\mathcal{G}=(V, E, g_V, c)$  be a BNPG, where $G=(V, E)$ is a network of players.
We first create the critical clique graph of~$G$ in polynomial time. Let~$T$ denote the critical clique graph of~$G$. For clarity, we call vertices in~$T$ nodes. For notational simplicity, we directly use the critical clique to denote its corresponding node in~$T$.
If~$G$ is disconnected, we run the following algorithm for each connected component. Then, we return the union of the profiles computed for all connected components if all subgames restricted to these connected components admit PSNEs; otherwise, we return ``{\no}''.
Therefore, let us assume now that~$G$ is connected, and hence~$T$ is a tree.
We choose any arbitrary node in~$T$ and make it as the root of the tree.
For each nonroot node~$K$ in~$T$, let~$K^{\text{P}}$ denote the parent node of~$K$ in~$T$. If~$K$ is the root, we define $K^{\text{P}}=\emptyset$. In addition, let~$\textsf{chd}(K)$ be the set of the children of~$K$ in~$T$ for every nonleaf node~$K$. If~$K$ is a leaf, we define ${\textsf{chd}}(K)=\{\emptyset\}$.
We use~$T_K$ to denote the subtree of~$T$ rooted at~$K$, and use~$\textsf{Ver}(T_K)$ to denote the set of vertices in~$G$ that are contained in the nodes of~$T_K$, i.e.,
\[\textsf{Ver}(T_K)=\bigcup_{K~\text{is a node in}~T_K} K.\]
For each node~$K$ in~$T$, we maintain a binary dynamic table ${\textsf{DT}}_K(x, y, z)$, where~$x$,~$y$, and~$z$ are integers such that
\begin{itemize}
\item $0\leq x\leq \abs{K}$,
\item $0\leq y\leq \abs{K^\text{P}}$, and
\item $0\leq z\leq \abs{\bigcup_{K'\in \textsf{chd}(K)}K'}$.
\end{itemize}
Particularly, $\textsf{DT}_K(x, y, z)$ is supposed to be~$1$ if and only if the subgame~$\mathcal{G}|_{{\textsf{Ver}}(T_{K^{\text{P}}})}$ admits a profile~${\bf{s}}$ under which (regard $T_{K^{\text{P}}}$ as~$T$ if~$K$ is the root of~$T$)
\begin{itemize}
\item everyone in~$K$ has exactly~$x$ closed neighbors in~$K$ who invest, i.e., $\abs{{\bf{s}}\cap K}=x$,
\item everyone in~$K$ has exactly~$y$ neighbors in~$K^{\text{P}}$ who invest, i.e., $\abs{{\bf{s}}\cap K^{\text{P}}}=y$,
\item everyone in $K$ has exactly~$z$ neighbors in their children nodes who invest, i.e., $\sum_{K'\in {\textsf{chd}}(K)}\abs{{\bf{s}}\cap K'}=z$ and,
\item none of players in ${\textsf{Ver}}(T_K)$ has an incentive to deviate.
\end{itemize}
Clearly, after computing all entries of the table~${\textsf{DT}}_{\hat{K}}$ associated to the root node~$\hat{K}$ in~$T$, we could conclude that the given BNPG admits a PSNE if and only if there if there is a $1$-valued entry ${\textsf{DT}}_{\hat{K}}(x, 0, z)=1$.

The tables associated to the nodes in~$T$ are computed in a bottom-up manner, from those associated to leaf nodes up to that associated to the root node.
Let ${\textsf{DT}}_K(x, y, z)$ be an entry considered at the moment. For each player~$v$, we define
\[\triangle(v, 0)=g_v(x+y+z)-c(v)-g_v(x+y+z-1),\]
%and
\[\triangle(v, 1)=g_v(x+y+z)-(g_v(x+y+z+1)-c(v)).\]
Note that if $\triangle(v, 0)<0$,~$v$ does not invest in any PSNE, and if $\triangle(v, 1)<0$,~$v$ must invest in every PSNE.
Therefore, if there is a player~$v\in K$ such that
$\triangle(v, 0)<0$ and $\triangle(v, 1)<0$, we immediately set $\textsf{DT}_K(x, y, z)=0$. Otherwise, we divide the players from~$K$ into
\begin{itemize}
\item $K_{{-}}=\{v\in K \setmid \triangle(v, 0)<0\}$;
\item $K_{{+}}=\{v\in K \setmid \triangle(v, 1)<0\}$; and
\item $K_*=K\setminus (K_{{+}}\cup K_{-})$.
\end{itemize}
We have the following observations.
\begin{itemize}
\item None of~$K_{{-}}$ invests in any PSNE;
\item All players in~$K_{{+}}$ must invest in all PSNEs;
\item Each player in~$K_*$ can be in both the set of investing players and the set of noninvesting players under all PSNEs.
\end{itemize}

Given the above observations, if $\abs{K_{{+}}}> x$ or $\abs{K_{{+}}\cup K_*}< x$, we directly set ${\textsf{DT}}_K(x, y, z)=0$.
Let us assume now that $\abs{K_{{+}}}\leq x$ and $\abs{K_{{+}}\cup K_*}\geq x$. We determine the value of ${\textsf{DT}}_K(x, y, z)$ as follows.
If~$K$ is a leaf node, we set ${\textsf{DT}}_K(x, y, z)=1$ if and only if $z=0$.
Otherwise, let~($K_1, K_2, \dots, K_t)$ be an arbitrary but fixed order of the children of~$K$ in~$T$, where~$t$ is the number of children of~$K$ in~$T$. Then, we set ${\textsf{DT}}_K(x, y, z)=1$ if and only if the following condition holds: there are entries $\textsf{DT}_{K_1}(x_1, y_1, z_1)=1$, $\textsf{DT}_{K_2}(x_2, y_2, z_2)=1$, $\dots$, $\textsf{DT}_{K_t}(x_t, y_t, z_t)=1$ such that $y_j=x$ for all $j\in [t]$ and $\sum_{j=1}^t x_t=z$.
In fact, in this case we let all players in~$K_{{+}}$ and arbitrarily $x-\abs{K_{{+}}}$ players in~$K_*$  invest; and let all the other players in~$K$ not invest. Importantly, the above condition can be checked in polynomial time by a dynamic programming algorithm. To this end, we maintain a binary dynamic table ${\textsf{DT}'}(i, x_i(1), x_i(2))$ where $i\in [t]$, and~$x_i(1)$ and~$x_i(2)$ are two integers such that $0\leq x_i(1)\leq \abs{K_i}$ and $0\leq x_i(2)\leq \abs{\bigcup_{j\in [i-1]}K_j}$ if $i>1$ and $x_i(2)=0$ if $i=1$.  In particular, ${\textsf{DT}'}(i, x_i(1), x_i(2))$ is supposed to be~$1$ if and only if there are entries $\textsf{DT}_{K_1}(x_1, y_1, z_1)=1$, $\textsf{DT}_{K_2}(x_2, y_2, z_2)=1$, $\dots$, $\textsf{DT}_{K_i}(x_i, y_i, z_i)=1$ such that $y_j=x$ for all $j\in [i]$, $\sum_{j=1}^{i-1} x_j=x_i(2)$, and $x_i=x_i(1)$.
The table is computed as follows. First, every base entry $\textsf{DT}'(1, x_1(1), x_1(2))$ has value~$1$ if and only if $x_1(2)=0$ and ${\textsf{DT}}_{K_1}(x_1(1), x, z')=1$ for some integer~$z'$. Then, the value of every entry ${\textsf{DT}'}(i, x_1(1), x_1(2))$ such that $i\geq 2$ is~$1$ if and only if there exists an entry ${\textsf{DT}}'(i-1, x_{i-1}(1), x_{i-2}(2))=1$ such that $x_{i-1}(1)\leq x_i(2)$, $x_{i-2}(2)=x_i(2)-x_{i-1}(1)$, and ${\textsf{DT}}_{K_i}(x_i(1), x, z')=1$ for some integer~$z'$. The above condition is satisfied if and only if $\textsf{DT}'(t, x_t(1), x_t(2))=1$ for some valid values of~$x_t(1)$ and~$x_t(2)$ such that $x_t(1)+x_t(2)=z$.

The algorithm can be implemented in polynomial time since for each node~$K$, the corresponding table has at most~$n^3$ entries, where~$n$ is the number of total players. So, we have in total at most~$n^4$ entries each of which can be computed in polynomial time.
\end{proof}

For {\prob{USWC}} and {\prob{ESWC}} we have similar results.

\begin{theorem}[*]
\label{thm-uswo-critical-clique-graph-tree-poly}
{\prob{USWC}} is polynomial-time solvable if the critical clique graph of the given network is a forest.
\end{theorem}

\begin{proof}
To prove the theorem, we derive a dynamic programming algorithm to solve {\prob{USWC}} in the case where the given network has a critical clique graph being a forest.
Let $\mathcal{G}=(V, E, g_V, c)$  be a BNPG, where $G=(V, E)$ is a network of players.
We first create the critical clique graph of~$G$ in polynomial time. Let~$T$ denote the critical clique graph of~$G$. For clarity, we call vertices in~$T$ nodes. For notational simplicity, we directly use the critical clique to denote its corresponding node in~$T$.
If~$G$ is disconnected, we run the following algorithm for each connected component. Then, we sum up all the values returned by the algorithms running on the connected components.
Therefore, let us assume now that~$G$ is connected, and hence~$T$ is a tree.
We choose any arbitrary node in~$T$ and make it as the root of the tree.
For each nonroot node~$K$ in~$T$, let~$K^{\text{P}}$ denote the parent node of~$K$ in~$T$. If~$K$ is the root, we define $K^{\text{P}}=\emptyset$. In addition, let~$\textsf{chd}(K)$ be the set of the children of~$K$ in~$T$ for every nonleaf node~$K$. If~$K$ is a leaf, we define ${\textsf{chd}}(K)=\{\emptyset\}$.
We use~$T_K$ to denote the subtree of~$T$ rooted at~$K$, and use~$\textsf{Ver}(T_K)$ to denote the set of vertices in~$G$ that are contained in the nodes in the subtree~$T_K$.
For each node~$K$ in~$T$, we maintain a dynamic table ${\textsf{DT}}_K(x, y, z)$, where~$x$,~$y$, and~$z$ are integers such that
\begin{itemize}
\item $0\leq x\leq \abs{K}$,
\item $0\leq y\leq \abs{K^\text{P}}$, and
\item $0\leq z\leq \abs{\bigcup_{K'\in \textsf{chd}(K)}K'}$.
\end{itemize}
We say that a profile of the subgame~$\mathcal{G}|_{\textsf{Ver}(T_{K^{\text{P}}})}$ (regard $T_{K^{\text{P}}}$ as~$T$ if~$K$ is the root of~$T$) is a $\textsf{DT}_K(x, y, z)$-compatible profile of the subgame if in this profile the following three conditions are satisfied:
\begin{enumerate}
\item exactly~$x$ players in~$K$ invest,
\item exactly~$y$ players in~$K^{\text{P}}$ invest, and
\item exactly~$z$ players in $\bigcup_{K'\in \textsf{chd}(K)}K'$ invest.
\end{enumerate}
The value of the entry ${\textsf{DT}}_K(x, y, z)$ is supposed to be the maximum possible USW of players in~${\textsf{Ver}}(T_K)$ under $\textsf{DT}_K(x, y, z)$-compatible profiles of the subgame~$\mathcal{G}|_{\textsf{Ver}(T_{K^{\text{P}}})}$.

The values of the entries in the table can be computed recursively, beginning from the leaf nodes up to the root node.
In particular, if ${\textsf{DT}}_K(x, y, z)$ is a leaf node in~$T$ (note that in this case~$z=0$), we compute ${\textsf{DT}}_K(x, y, 0)$ as follows.
For every player~$v\in K$, the number of closed neighbors who invest is exactly~$x+y$ in every ${\textsf{DT}}_K(x, y, 0)$-compatible profile. Then, the utility of every player~$v\in K$ to investing and not to investing are respectively $g_v(x+y)-c(v)$ and $g_v(x+y)$. We order players in~$K$ according to a nondecreasing order of the investment costs~$c(v)$ of players~$v\in K$. Then, it is easy to see that a ${\textsf{DT}}_K(x, y, 0)$-compatible profile which consists of the first~$x$ players in the order achieves the maximum possible USW of players in~$K$, among all ${\textsf{DT}}_K(x, y, 0)$-compatible profiles of the game restricted to $\textsf{Ver}(T_{K^{\text{P}}})$. Let~$H$ be the set of the first~$x$ players in the order. In light of the above discussion, we define
\begin{align*}
{\textsf{DT}}_K(x, y, 0)& =\sum_{v\in H}(g_v(x+y)-c(v))+\sum_{v\in K\setminus H}g_v(x+y)\\
&= \sum_{v\in K} g_v(x+y)-\sum_{v\in H}c(v).\\
\end{align*}

For a nonleaf node~$K$, we compute ${\textsf{DT}}_K(x, y, z)$ as follows, assuming that the values of all tables associated to the descendants of~$K$ are already computed.
First, similar to the above case, we first order players in~$K$ according to a nondecreasing order of~$c(v)$, $v\in K$, and let~$H$ denote the first~$x$ players in the order. Then, we define
\[s=\sum_{v\in H}\left(g_v(x+y+z)-c(v)\right)+\sum_{v\in K\setminus H}g_v(x+y+z),\]
which is the maximum possible USW of players in~$K$ under ${\textsf{DT}}_K(x, y, z)$-compatible profiles. However, we need also to take into account the USW of the descendents of~$K$. We solve this by a dynamic programming. Let $(K_1, K_2, \dots, K_t)$ be an arbitrary but fixed order of children of~$K$, where~$t$ is the number of children of~$K$. As~$K$ is the parent of each~$K_i$, only the entries ${\textsf{DT}}_{K_i}(x_i, y_i, z_i)$ such that $y_i=x$ are relevant to our computation. For each $i\in [t]$, we maintain a dynamic table ${\textsf{DT}'}(i, x_i(1), x_i(2))$ where $x_i(1)$ and $x_i(2)$ are two integers such that $0\leq x_i(1)\leq \abs{K_i}$ and $0\leq x_i(2)\leq \bigcup_{j\in [i-1]}K_j$ if $i>1$ and $x_i(2)=0$ if $i=1$. The integer $x_i(1)$ and $x_i(2)$ respectively indicate the number of players in~$K_i$ who invest and the number of players in $\bigcup_{j\in [i-1]}K_j$ who invest, and the value of the entry is the maximum possible USW of players in $\bigcup_{j\in [i]}K_j$ and their descendants, i.e., the USW of players in $\bigcup_{j\in [i]} {\textsf{Ver}}(T_{K_j})$, under the above restrictions.
The table is computed as follows. First, we let
\[{\textsf{DT}'}(1, x_1(1), 0)=\max_{z'}{\textsf{DT}}_{K_1}(x_1(1), x, z'),\]
where~$z'$ runs over all possible values.
Then, for each~$i$ from~$2$ to~$t$ (this applies only when $t\geq 2$), the entry ${\textsf{DT}}'(i, x_i(1), x_i(2))$ is computed by the following recursive:
\begin{align*}
\textsf{DT}'(i, x_i(1), x_i(2))=\max_{z'}{\textsf{DT}}_{K_i}(x_i(1), x, z')+\\
\max_{\substack{0\leq x_{i-1}(1)\leq \abs{K_{i-1}} \\ x_i(2)-x_{i-1}(1)\geq 0}}\{\textsf{DT}'(i-1, x_{i-1}(1), x_i(2)-x_{i-1}(1))\},\\
\end{align*}
where~$z'$ runs all possible values.
After all the entries are updated, we define
\[s'=\max_{\substack{0\leq x_t(1)\leq \abs{K_t}\\ x_t(1)+x_t(2)=z}}\{\textsf{DT}'(t, x_t(1), x_t(2))\}.\]
Now we are ready to update the entry for~$K$.
In particular, we define
\[\textsf{DT}_K(x, y, z)=s+s'.\]
After the entries of the table~$\textsf{DT}$ are computed, we return
\[\max_{x', z'}\{\textsf{DT}_K(x', 0, z')\},\]
where~$K$ is the root node and~$x'$ and~$z'$ run over all possible values.

The algorithm can be implemented in polynomial time since for each node~$K$, the corresponding table has at most~$n^3$ entries, where~$n$ is the number of total players. So, we have in total at most~$n^4$ entries each of which can be computed in polynomial time.
\end{proof}

\begin{theorem}[*]
{\prob{ESWC}} is polynomial-time solvable if the critical clique graph of the given network is a forest.
\end{theorem}

%\onlyfull
{
\begin{proof}
The algorithm is similar to the one in the proof of Theorem~\ref{thm-uswo-critical-clique-graph-tree-poly}. In particular, we guess the ESW of the desired profile. There can be polynomially many guesses. For each guessed value~$q$, we determine if there is a profile of ESW at least~$q$, i.e.,  every player receives utility at least~$q$ under this profile. This can be solved using a dynamic programming algorithm.
We adopt the same notations in the proof of Theorem~\ref{thm-uswo-critical-clique-graph-tree-poly}.
However, in the current algorithm, each entry $\textsf{DT}_K(x, y, z)$ in the dynamic tables takes only binary values.
Precisely, $\textsf{DT}_K(x, y, z)$ is supposed to be~$1$  if and only if the subgame~$\mathcal{G}|_{{\textsf{Ver}}(T_K^{\text{P}})}$ admits a profile~${\bf{s}}$ under which (regard $T_{K^{\text{P}}}$ as~$T$ if~$K$ is the root of~$T$)
\begin{enumerate}
\item exactly~$x$ players in~$K$ invest, i.e., $\abs{{\bf{s}}\cap K}=x$,
\item exactly~$y$ players in~$K^{\text{P}}$ invest, i.e., $\abs{{\bf{s}}\cap K^{\text{P}}}=y$,
\item exactly~$z$ players in $\bigcup_{K'\in \textsf{chd}(K)}K'$ invest, i.e., $\sum_{K'\in \textsf{chd}(K)}\abs{{\bf{s}}\cap K'}=z$ and, more importantly,
\item every player $v\in {\textsf{Ver}}(T_K)$ obtains utility at least~$q$ under this profile, i.e., $\mu(v, {\bf{s}}, \mathcal{G}|_{{\textsf{Ver}}(T_{K^{\text{P}}})})\geq q$.
\end{enumerate}
The values of entries in the tables associated to leaf nodes can be computed trivially based on the above definition. We describe how to update the remaining tables.
Let ${\textsf{DT}}_K(x, y, z)$ be the currently considered entry in a table associated to a node~$K$ in~$T$. Let $K_1, K_2, \dots, K_t$ be the children of~$K$ in~$T$. We set ${\textsf{DT}}_K(x, y, z)$ to be~$1$ if and only if the following conditions hold simultaneously.
\begin{enumerate}
\item There is a subset $K'\subseteq K$ of cardinality~$x$ such that
\begin{itemize}
\item for every $v\in K'$ it holds that $g_v(x+y+z)-c(v)\geq q$; and
\item for every $v\in K\setminus K'$ it holds that $g(x+y+z)\geq q$.
\end{itemize}
\item There are ${\textsf{DT}}_{K_1}(x_1, x, z_1)=1$, ${\textsf{DT}}_{K_2}(x_2, x, z_2)=1$, $\dots$, ${\textsf{DT}}_{K_t}(x_t, x, z_t)=1$ such that $\sum_{j=1}^t x_j=z$.
\end{enumerate}
We point out that both of the above two conditions can be checked in polynomial time.
Precisely, to check the first condition, we define $A=\{v\in K \setmid g_v(x+y+z)< q\}$ and $B=\{v\in K \setmid g_v(x+y+z)\geq q, g_v(x+y+z)-c(v)< q\}$, both of which can be computed in polynomial time. Clearly, if $A\neq \emptyset$, Condition~1 does not hold. Otherwise, if $\abs{B}> \abs{K}-x$, we also conclude that Condition~1 does not hold, because due to the definition of~$B$, none of them should invest in order to obtain utility at least~$q$. If none of the above two cases occurs, we conclude that Condition~1 holds. As a matter of fact, in this case, we can let~$K'$ be any subset of $K\setminus B$ of cardinality~$x$.
To check Condition~2, we use a similar dynamic programming algorithm with the associated table ${\textsf{DT}'}(i, x_i(1), x_i(2))$ in the proof of Theorem~\ref{thm-uswo-critical-clique-graph-tree-poly}.

The algorithm runs in polynomial time since there are polynomially many entries and computing the value for each entry can be done in polynomial time.
\end{proof}
}

\section{Networks with Bounded Treewidth}
In this section, we study another prevalent class of tree-like networks, namely, the networks with a constant bounded treewidth. We show that the problems studied in the paper are polynomial-time solvable in this special case.
Notice that as every clique of size~$k$ has treewidth~$k-1$, the results established in the previous section do not cover the polynomial-time solvability in this case. The other direction does not hold too because every cycle has treewidth three but the critical clique graph of every cycle is itself.
%
%Arguably, the notion of tree-width, first introduced by Robertson and Seymour~\cite{DBLP:journals/jal/RobertsonS86}, has been one of the most influential structural parameters that has received a tremendous amount of attention. Generally speaking this notion measures how close a graph is to a tree. Many important graph classes such as trees, series-parallel, outerplaner graphs, etc., have constant treewidths.  Below we elaborate on this notion formally.

The following notion is due to~\cite{DBLP:journals/jal/RobertsonS86}.

A {\it{tree decomposition}} of a graph $G=(V, E)$ is a tuple $(T=(L, F),\mathcal{B})$, where~$T$ is a rooted tree with vertex set~$L$ and edge set~$F$,
and $\mathcal{B}=\{B_x \subseteq V \mid x\in L\}$ is a
collection of subsets of vertices of~$G$ such that the following three conditions are satisfied:
\begin{itemize}
\item every $v\in V$ is contained in at least one element of $\mathcal{B}$; %For each vertex $v\in V$ there exists at least one $B_x\in \mathcal{B}$ such that $v\in B_x$, i.e., every vertex of~$G$ is in at least one subset in~$\mathcal{B}$;
\item for each edge $\edge{v}{v'}\in E$, there exists at least one $B_x\in \mathcal{B}$ such that $v, v'\in B_x$; and %For each edge $\edge{v}{u}\in E$ in~$G$ there exists at least one $B_x\in \mathcal{B}$ such that $v,u\in B_x$, i.e., every edge of~$G$ is contained in at least one subset of~$\mathcal{B}$; and
\item for every $v\in V$, if~$v$ is in two distinct $B_x, B_y\in \mathcal{B}$, then~$v$ is in every $B_z\in \mathcal{B}$ where~$z$ is on the unique path between~$x$ and~$y$ in~$T$.
\end{itemize}
The {\it{width}} of the tree decomposition is $\max_{B\in \mathcal{B}}{|B|-1}$.
The {\it{treewidth}} of a graph~$G$, denoted by~$\omega(G)$, is the width of a tree decomposition of~$G$ with the minimum width.
The subsets in~$\mathcal{B}$ are often called {\it{bags}}. The root bag in the decomposition is the bag associated to the root of~$T$.
To avoid confusion, in the following we refer to the vertices of~$T$ as nodes.
The parent bag of a bag $B_i\in \mathcal{B}$ means the bag associated to the parent of~$i$.

A more refined notion is the so-called nice tree decomposition.
Particular, a {\it{nice tree decomposition}} $(T, \mathcal{B})$ of a graph~$G$ is a specific tree decomposition of~$G$ satisfying the following conditions:
\begin{itemize}
\item every bag associated to the root or a leaf of~$T$ is empty;%, i.e., $B_x=\emptyset$ if~$x$ is the root or a leaf in~$T$.
\item inner nodes of~$T$ are categorized into {\it{introduce nodes, forget nodes}}, and {\it{join nodes}} such that
\begin{itemize}
\item each introduce node~$x$ has exactly one child~$y$ such that $B_y\subset B_x$ and $|B_x\setminus B_y|=1$; %, i.e.,~$B_x$ has exactly one more element than~$B_y$.
\item each forget node~$x$ has exactly one child~$y$ such that $B_x\subset B_y$ and $|B_y\setminus B_x|=1$; and%, i.e.,~$B_x$ is obtained from~$B_y$ by removing one element.
\item each join node~$x$ has exactly two children~$y$ and~$z$ such that $B_x=B_y=B_z$.
\end{itemize}
\end{itemize}

For ease of exposition, we sometimes call a bag associated to a join (resp.\ forget, introduce) node a join (resp.\ forget, introduce) bag.
It can be known from the definition that in a nice tree decomposition of a graph~$G$, every vertex in~$G$ can be introduced several times but can be only forgotten once.
Nice tree decomposition was introduced by Bodlaender and Kloks~\shortcite{DBLP:conf/icalp/BodlaenderK91}, and has been used in tacking many problems.
 At first glance, nice tree decomposition seems very restrictive.
 However, it is proved that given a tree decomposition of width~$\omega$, one can calculate a nice tree decomposition with the same width in polynomial-time~\cite{DBLP:books/sp/Kloks94}.

%\begin{lemma}[\cite{DBLP:books/sp/Kloks94} (Lemma~13.1.3)]
%\label{lem-nice-tree-decomposition}
%Given a tree decomposition of width~$\omega$ of a graph $G=(N, A)$, one can calculate a nice tree decomposition of~$G$ of width~$\omega$ in polynomial-time.
%Moreover, the number of nodes of the nice tree decomposition is at most $4\cdot |N|$.
%\end{lemma}

It is known that calculating the treewidth of a graph~$G$ is {\nph}~\cite{DBLP:journals/actaC/Bodlaender93}.
However, determining whether a graph has a constant bounded treewidth can be solved in polynomial time and, moreover, powerful heuristic algorithms, approximation algorithms, and fixed-parameter algorithms to calculating treewidth have been reported~\cite{DBLP:conf/birthday/Bodlaender12,DBLP:journals/siamcomp/BodlaenderDDFLP16,DBLP:conf/iwpec/ZandenB17}.
Hence, in the following results we assume that a nice tree-decomposition of the given network is given.

\begin{theorem}[*]
{\prob{PSNEC}} is polynomial-time solvable if the treewidth of the given network is a constant.
\end{theorem}

%\onlyfull
{
\begin{proof}
Let $\mathcal{G}=(V, E, g_V, c)$ be a BNPG, where~$(V, E)$ is a network of players,~$g_V$ is a set of externality functions of players in~$V$, one for each player, and $c: V\rightarrow \mathbb{R}_{\geq 0}$ is the investment cost function. For every player $v\in V$, let $g_v: \mathbb{N}_{0}\rightarrow \mathbb{R}_{\geq 0}$ denote its externality function in~$g_V$. Let~$G$ denote the network $(V, E)$, and  let $n=\abs{V}$ be the number of players. In addition, let $(T, \mathcal{B})$ be a nice-tree decomposition of~$G$ which is of polynomial size in~$n$ and of width at most~$p$ for some constant~$p$. For a node~$i$ in~$T$, let $B_i\in \mathcal{B}$ denote its associated bag in the nice tree decomposition. Moreover, let~$T_i$ denote the subtree of~$T$ rooted at~$i$, let~$G_i=G[\bigcup_{j\in \ver{T_i}} B_j]$ denote the subgraph of~$G$ induced by all vertices contained in bags associated to nodes in~$T_i$, and let~$V_i$ denote the vertex set of~$G_i$. For each nonroot bag $B_i\in \mathcal{B}$, we define~$B_i^{\text{P}}$ as the parent bag of~$B_i$. If~$B_i$ is the root bag, we define $B_i^{\text{P}}=\emptyset$.
For each bag~$B_i$ associated to a node~$i$ and each vertex~$v\in B_i$, we use~$n_i(v)$ to denote the number of neighbors of~$v$ in the subgraph~$G_i-B_i$.
In the following, we derive a dynamic programming algorithm to determine if the BNPG game $\mathcal{G}$ has a PSNE profile, and if so, the algorithm returns a PSNE profile.

For each bag~$B_i\in \mathcal{B}$, we maintain a binary dynamic table ${\textsf{DT}}_i(U, f)$ where~$U$ runs over all subsets of~$B_i$ and~$f$ runs over all functions $f: B_i\rightarrow \mathbb{N}_0$ such that~$f(u)\leq n_i(u)$ for every~$u\in B_i$.
The entry~${\textsf{DT}}_i(U, f)$ is supposed to be~$1$ if the subgame $\mathcal{G}|_{V_i}$ admits a profile~${\bf{s}}$ such that
\begin{enumerate}
\item ${\bf{s}}\cap B_i=U$;
\item no player in $V(G_i)\setminus B_{t'}$ has an incentive to deviate under~${\bf{s}}$; and
\item every player $v\in B_i$ has exactly~$f(v)$ investing neighbors contained in $G_i-B_i$ under~${\bf{s}}$.
\end{enumerate}
%In the collection $\mathcal{C}(B)$ stands for a partial BNPG partition $(V', B\setminus V')$.

We compute the tables for the bags in a bottom-up manner, beginning from those associated to leaf nodes to the table associated to the root. Assume that ${\textsf{DT}}_i(U, f)$ is the currently considered entry. To compute the value of this entry, we distinguish the following cases.

\begin{description}
\item[Case~1: $B_i^{\text{P}}$ is a join or an introduce bag, or~$i$ is the root of~$T$] \hfill

We further distinguish between the following subcases.

\begin{description}
\item[Case~1.1:~$B_i$ is a join bag] \hfill

Let~$x$ and~$y$ be the two children of~$i$. Therefore, it holds that $B_x=B_y=B_i$. In this case, we set ${\textsf{DT}}_i(U, f)=1$ if and only if ${\textsf{DT}}_x(U, f)=1$ and ${\textsf{DT}}_y(U, f)=1$.

\item[Case~1.2:~$B_i$ is an introduce bag] \hfill

Let~$x$ be the child of~$i$, and let~$B_i\setminus B_x=\{v\}$. (In this case,~$i$ cannot the root of~$T$) Notice that in this case,~$v$ does not have any neighbor in $G_i-B_i$. Then, we set ${\textsf{DT}}_i(U, f)=1$ if and only if $f(v)=0$ and ${\textsf{DT}}_x(U\setminus \{v\}, f|_{\neg v})=1$.

\item[Case~1.3:~$B_i$ is a forget node] \hfill

Let~$x$ be the child of~$i$, and let~$B_x\setminus B_i=\{v\}$. Then, we set ${\textsf{DT}}_i(U, f)=1$ if and only if there exists an entry ${\textsf{DT}}_x(U', f')=1$ such that $U'\setminus \{v\}=U$ and $f'|_{\neg v}=f$.
\end{description}

\item[Case~2: $B_i^{\text{P}}$ is a forget node] \hfill

Let $B_i\setminus B_i^{\text{P}}=\{v\}$.
We further divided into three subcases.
\begin{description}
\item[Case~2.1 $B_i$ is a join bag]\hfill

Let~$x$ and~$y$ denote the two children of~$i$ in~$T$. Then, we set ${\textsf{DT}}_i(U, f)=1$ if and only if ${\textsf{DT}}_x(U, f)={\textsf{DT}}_y(U, f)=1$ and, moreover, $g_v(n_G[v, U]+f(v))-c(v)\geq g_v(n_G(v, U))$ when $v\in U$ and $g_v(n_G(v, U)+f(v))\geq g_v(n_G(v, U)+f(v))-c(v)$ when $v\in B_i\setminus U$.

\item[Case~2.2 $B_i$ is an introduce bag] \hfill

Let~$x$ denote the child of~$i$ in~$T$. Let $B_i\setminus B_x=\{u\}$. Obviously,~$u$ does not have any neighbor in $G_i-B_i$. Therefore, if $f(u)>0$, we directly set ${\textsf{DT}}_i(U, f)=0$. Let us assume now that $f(u)=0$. Then, we set ${\textsf{DT}}_i(U, f)=1$ if and only if ${\textsf{DT}}_x(U\setminus \{u\}, f|_{\neg u})=1$ and, moreover, it holds that
$g_v(n_G[v, U]+f(v))-c(v)\geq g_v(n_G(v, U))$ when $v\in U$ and
$g_v(n_G(v, U))\geq g_v(n_G[v, U]+f(v))-c(v)$ when $v\in B_i\setminus U$.

\item[Case 2.3, $B_i$ is a forget bag] \hfill

Let~$x$ denote the child of~$i$ in~$T$, and let $B_i\setminus B_x=\{u\}$.
In this case, we set ${\textsf{DT}}_i(U, f)=1$ if and only if the following two conditions holds:
\begin{itemize}
\item there exists an entry ${\textsf{DT}}_x(U', f')=1$ such that $U'\setminus \{u\}=U$, $f'|_{\neg u}=f$ and
\item $g_v(n_G[v, U]+f(v))-c(v)\geq g_v(n_G(v, U))$ when $v\in U$ and
$g_v(n_G(v, U))\geq g_v(n_G[v, U]+f(v))-c(v)$ when $v\in B_i\setminus U$.
\end{itemize}
\end{description}
\end{description}

Recall that the root bag is empty. Therefore, there is only one entry in the table associated to the root. The game~$\mathcal{G}$ admits a PSNE profile if the only entry in the table of the root takes the value~$1$. A PSNE can be constructed using standard backtracking technique of dynamic programming algorithms if such a profile exists.

Finally, we analysis the running time of the algorithm. First, there are in total~$\bigo(n)$ bags where~$n$ denotes the number of all players in the given game. Let~$k$ be the treewidth of the given nice tree decomposition of the given network. For each bag~$B$, the set associated to~$B$ is of cardinality $2^{|B|}\cdot n^{|B|}= \bigos(2n^k)$, which can be constructed in~$\bigos(2n^k)$ time. Therefore, the running time of the algorithm is bounded by~$\bigos(2n^k)$, which is polynomial if~$k$ is a constant.
\end{proof}
}

\onlyfull{
A cluster is a disjoint of cliques. It is also called the~$P_3$-free graphs and 2-leaf powers.
The {\it{distance to a cluster}} of a graph is the minimum number of vertices need to be deleted to obtain a cluster.
We extend the result for clique in the following way.

\begin{theorem}
With respect to the distance to a cluster, all problems become {\fpt}.
\end{theorem}

{\bf{Remark.}} The above three NP-hardness results hold when the network has diameter at most 3.
}

\begin{theorem}[*]
{\prob{USWC}} is polynomial-time solvable if the treewidth of the given network is a constant.
\end{theorem}

%\onlyfull
{
\begin{proof}
Let $\mathcal{G}=(V, E, g_V, c)$ be a BNPG, where~$(V, E)$ is a network of players,~$g_V$ is a set of externality functions of players in~$V$, one for each player, and $c: V\rightarrow \mathbb{R}_{\geq 0}$ is the cost function. For every player $v\in V$, let $g_v: \mathbb{N}_{0}\rightarrow \mathbb{R}_{\geq 0}$ denote its externality function in~$g_V$. Let~$G$ denote the network $(V, E)$, and  let $n=\abs{V}$ be the number of players. In addition, let $(T, \mathcal{B})$ be a nice-tree decomposition of~$G$ which is of polynomial size in~$n$ and of width at most~$p$ for some constant~$p$. For a node~$i$ in~$T$, let $B_i\in \mathcal{B}$ denote its associated bag. Moreover, let~$T_i$ denote the subtree of~$T$ rooted at~$i$, let~$G_i=G[\bigcup_{j\in \ver{T_i}} B_j]$ denote the subgraph of~$G$ induced by all vertices contained in bags associated to nodes in~$T_i$, and let~$V_i$ denote the vertex set of~$G_i$. For each nonroot bag $B_i\in \mathcal{B}$, we define~$B_i^{\text{P}}$ as the parent bag of~$B_i$. If~$B_i$ is the root bag, we define $B_i^{\text{P}}=\emptyset$.
For each bag~$B_i$ associated to a node~$i$ and each vertex~$v\in B_i$, we use~$n_i(v)$ to denote the number of neighbors of~$v$ in the subgraph~$G_i-B_i$.
In the following, we derive a dynamic programming algorithm to compute a profile of~$\mathcal{G}$ with the maximum possible USW.

For each bag~$B_i\in \mathcal{B}$, we maintain a dynamic table ${\textsf{DT}}_i(U, f)$ where~$U$ runs over all subsets of~$B_i$ and~$f$ runs over all functions $f: B_i\rightarrow \mathbb{N}_0$ such that~$f(u)\leq n_i(u)$ for every~$u\in B_i$.
A profile~${\bf{s}}$ of a the subgame $\mathcal{G}|_{V_i}$ is consistent with the tuple $(U, f)$ if among the players in~$B_i$ exactly those in~$U$ invest and, moreover, each $v\in B_i$ has exactly~$f(v)$ neighbors in $G_i- B_i$ who invest, i.e., ${\bf{s}}\cap B_i=U$ and $n_G(v, {\bf{s}}\setminus U)=f(v)$.
The entry ${\textsf{DT}}_i(U, f)$ is defined as
\[\max_{{\bf{s}}} \left(\sum_{v\in V_i\setminus B_i^{\text{P}}} \mu(v, \mathcal{G}|_{V_i}, {\bf{s}})\right),\]
where~${\bf{s}}$ runs over all profiles of the subgame $\mathcal{G}|_{V_i}$ that are consistent with $(U, f)$. If the subgame does not admit a profile which is consistent with $(U, f)$, ${\textsf{DT}}_i(U, f)=-\infty$.

To compute the values of the entries, we distinguish the following cases. Let ${\textsf{DT}}_i(U, f)$ denote the currently considered entry.

\begin{description}
\item[Case~1:~$B_i^{\text{P}}$ is a join or an introduce bag, or~$i$ is the root of~$T$] \hfill

We consider the following subcases.
\begin{description}
\item[Case 1.1:~$B_i$ is a join bag] \hfill

Let~$x$ and~$y$ denote the two children of~$i$ in the tree~$T$. Then, we define
\[{\textsf{DT}}_i(U, f)={\textsf{DT}}_x(U, f)+{\textsf{DT}}_y(U, f).\]

\item[Case~1.2 $B_i$ is an introduce bag] \hfill

Let~$x$ be the child of~$i$ in~$T$, and let $B_i\setminus B_x=\{v\}$. (In this case,~$i$ cannot be the root of~$T$.) Observe that in this case~$v$ does not have any neighbors in $G_i-B_i$. Hence, if $f(v)>0$, we let ${\textsf{DT}}_i(U, f)=-\infty$. Otherwise, it must be that $f(v)=0$, and we let
\[{\textsf{DT}}_i(U, f)={\textsf{DT}}_x(U\setminus \{v\}, f|_{\neg v}).\]

\item[Case 1.3: $B_i$ is a forget bag] \hfill

Let~$x$ be the child of~$i$ in~$T$ and let $B_x\setminus B_i=\{v\}$. Then, we define
\[{\textsf{DT}}_x(U, f)=\max_{\substack{U'\subseteq B_x~\text{s.t.}~U'\setminus \{v\}=U,\\ f':B_x\rightarrow \mathbb{N}_{0}~\text{s.t.}~f'|_{\neg v}=f}} \{{\textsf{DT}}_x(U', f')\}.\]
\end{description}
%where
%\begin{equation*}
%\mu(v)=
%\begin{cases}
%g_v(1+f(v)+n_G(v, U))-c(v) & v\in U\\
%g_v(f(v)+n_G(v, U))-c(v) & v\in B_i\setminus U\\
%\end{cases}
%\end{equation*}

\item[Case~2:~$B_i^{\text{P}}$ is a forget bag] \hfill

Let $B_i\setminus B_i^{\text{P}}=\{v\}$. We consider the following subcases.
\begin{description}
\item[Case 2.1: $B_i$ is a join bag] \hfill

Let~$x$ and~$y$ denote the two children of~$i$ in~$T$. We define
\[{\textsf{DT}}_i(U, f)={\textsf{DT}}_x(U, f)+{\textsf{DT}}_y(U, f)+\pi(v),\]
where
\begin{equation*}
\pi(v)=
\begin{cases}
g_v(f(v)+n_G[v, U])-c(v) & v\in U\\
g_v(f(v)+n_G(v, U)) & v\in B_i\setminus U\\
\end{cases}
\end{equation*}

\item[Case 2.2: $B_i$ is an introduce bag] \hfill

Let~$x$ denote the child of~$i$ in~$T$. In addition, let $B_i\setminus B_x=\{u\}$. In this case,~$u$ does not have any neighbor in $G_i-B_i$, and hence if~$f(u)>0$ the game restricted to~$V_i$ does not have a profile which is consistent with $(U, f)$. Therefore, if $f(u)>0$ we directly set ${\textsf{DT}}_i(U, f)=-\infty$. Let us assume that $f(u)=0$. Then, we define
\[{\textsf{DT}}_x(U, f)={\textsf{DT}}_x(U\setminus \{u\}, f|_{\neg u})+\pi(v),\]
where~$\pi(v)$ is as defined in Case~2.1.

\item[Case 2.3: $B_i$ is a forget bag] \hfill

Let~$x$ denote the child of~$i$ in~$T$. In addition, let $B_x\setminus B_i=\{u\}$.
Let~$\pi(v)$ be as defined in Case~2.1. Then, we define
\end{description}
\[{\textsf{DT}}_x(U, f)=\pi(v)+\max_{\substack{U'\subseteq B_x~\text{s.t.}~U'\setminus \{u\}=U\\ f': B_i\rightarrow \mathbb{N}_{0}~\text{s.t.}~f'|_{\neg u}=f}}\{{\textsf{DT}}_x(U', f')\}.\]
\end{description}

The algorithm returns ${\textsf{DT}}_r(\emptyset, \emptyset)$, where~$r$ is the root. The dynamic programming algorithm runs in polynomial time as there are in total polynomially many entries and computing the value of each entry takes polynomial time.
\end{proof}
}

\begin{theorem}
{\prob{ESWC}} is polynomial-time solvable if the treewidth of the given network is a constant.
\end{theorem}

\begin{proof}
Let $\mathcal{G}=(V, E, g_V, c)$ be a BNPG, where~$(V, E)$ is a network of players,~$g_V$ is a set of externality functions of players in~$V$, one for each player, and $c: V\rightarrow \mathbb{R}_{\geq 0}$ is the cost function. For every player $v\in V$, let $g_v: \mathbb{N}_{0}\rightarrow \mathbb{R}_{\geq 0}$ denote its externality function in~$g_V$. Let~$G$ denote the network $(V, E)$, and  let $n=\abs{V}$ be the number of players. In addition, let $(T, \mathcal{B})$ be a nice-tree decomposition of~$G$ which is of polynomial size in~$n$ and of width at most~$p$ for some constant~$p$. For a node~$i$ in~$T$, let $B_i\in \mathcal{B}$ denote its associated bag. Moreover, let~$T_i$ denote the subtree of~$T$ rooted at~$i$, and let~$G_i=G[\bigcup_{j\in \ver{T_i}} B_j]$ denote the subgraph of~$G$ induced by all vertices contained in bags associated to nodes in~$T_i$. For each nonroot bag $B_i\in \mathcal{B}$, we define~$B_i^{\text{P}}$ as the parent bag of~$B_i$. If~$B_i$ is the root bag, we define $B_i^{\text{P}}=\emptyset$.

For each bag~$B_i\in \mathcal{B}$, where $i\in \ver{T}$, and each vertex~$v\in B_i$, we use~$n_i(v)$ to denote the number of neighbors of~$v$ in the subgraph~$G_i-B_i$.
We derive an algorithm as follows.
The algorithm first guesses the ESW of the desired profile.
Note that we need only to consider at most $2n\cdot (n+1)$ possible values (the utility of every player has at most~$2(n+1)$ possible values and there are~$n$ players).
For each guessed value~$q$ of the ESW, we solve the problem of determining if~$\mathcal{G}$ admits a profile of ESW at least~$q$. This problem can be solved by the following dynamic programming algorithm running in polynomial time.

For each bag~$B_i\in \mathcal{B}$, we maintain a binary dynamic table $\textsf{DT}_i(U, f)$, where~$U$ runs over all subsets of~$B_i$ and~$f$ runs over all functions $f: B_i\rightarrow \mathbb{N}_0$ such that~$f(u)\leq n_i(u)$ for every~$u\in B_i$. Each entry ${\textsf{DT}}_i(U, f)$ is supposed to be~$1$ if and only if~$\mathcal{G}$ admits a profile such that
\begin{enumerate}
\item[(1)] among the players in~$B_i$, exactly those in~$U$ invest;
\item[(2)] every $v\in B_i$ has exactly~$f(v)$ neighbors in $G_i- B_i$ who invest; and
\item[(3)] everyone in~$B_i$ but not in~$B_i^{\text{P}}$ obtains utility at least~$q$, i.e.,
\begin{enumerate}
\item[(3.1)] for every $v\in U\setminus B_i^{\text{P}}$, it holds that $g_v(n_G[v, U]+f(v))-c(v)\geq q$; and
\item[(3.2)] for every $v\in (B_i\setminus U)\setminus B_i^{\text{P}}$, it holds that $g_v(n_G(v, U)+f(v))\geq q$.
\end{enumerate}
\end{enumerate}
We do not request players in~$B_i$ who also appear~$B_i^{\text{P}}$ to obtain the threshold utility~$q$ at the moment because we do not have the complete information over the number of their neighbors who invest at the moment. These players  will be treated when they leave their corresponding forget bags.
The tables are computed in a bottom-up manner, from those maintained for the leaf nodes up to that for the root node in~$T$. Assume that~$i$ is the currently considered node. We show how to compute $\textsf{DT}_i(U, f)$ as follows. Recall that every leaf bag is empty. So, if~$i$ is a leaf node, the table for~$i$ contains only one entry with the two parameters being an empty set and an empty function. We let this entry contain the value~$1$. Let us assume now that~$i$ is not a leaf node. We distinguish the following cases.

%Similar to the proof of the above theorem, we consider the following cases.
\begin{description}
\item[Case:~$B_i$ is a join bag] \hfill

Let~$x$ and~$y$ denote the two children of~$i$ in the tree~$T$.
If~$i$ is the root, or~$B_i^{\text{P}}$ is a join bag or an introduce bag, we set ${\textsf{DT}}_i(U, f)=1$ if and only if ${\textsf{DT}}_x(U, f)={\textsf{DT}}_y(U, f)=1$.
If~$B_i^{\text{P}}$ is a forget bag, let $B_i\setminus B_i^{\text{P}}=\{v\}$. Then, we set ${\textsf{DT}}_i(U, f)=1$ if and only if ${\textsf{DT}}_x(U, f)={\textsf{DT}}_y(U, f)=1$  and one of the following holds:
\begin{itemize}
\item $v\in U$ and $g_v(n_G[v, U]+f(v))-c(v)\geq q$; or
\item $v\in B_i\setminus U$ and $g_{v}(n_G(v, U)+f(v))\geq q$.
\end{itemize}

\item[Case:~$B_i$ is an introduce bag] \hfill

Let~$x$ be the child of~$i$ in~$T$.
Let $B_i\setminus B_x=\{v\}$. Note that~$v$ does not have any neighbor in $G_i-B_i$. Thus, if $f(v)>0$, we directly set ${\textsf{DT}}_i(U, f)$=0. We assume now that $f(v)=0$. We consider the following cases. First, if~$B_i^{\text{P}}$ is a join bag or an introduce bag, then ${\textsf{DT}}_i(U, f)=1$ if and only if ${\textsf{DT}}_x(U\setminus \{v\}, f|_{\neg \{v\}})=1$.
If~$B_i^{\text{P}}$ is a forget bag, let $B_i\setminus B_i^{\text{P}}=\{u\}$. % We do the following.
%If $u=v$, then we set ${\textsf{DT}}_i(U, f)=1$ if and only if ${\textsf{DT}}_x(U\setminus \{v\}, f|_{\neg \{v\}})=1$ and one of the following holds:
%\begin{itemize}
%\item $u\in U$ and $g_u(n_G[u, U])-c(u)\geq q$;
%\item $u\in B_i\setminus U$ and $g_u(n_G(u, U))\geq q$.
%\end{itemize}
%If $u\neq v$, then
Then, we set ${\textsf{DT}}_i(U, f)=1$ if and only if ${\textsf{DT}}_x(U\setminus \{v\}, f|_{\neg \{v\}})=1$ and one of the following holds (notice that when $u=v$ we have $f(u)=0$):
\begin{itemize}
\item $u\in U$ and $g_u(n_G[u, U]+f(u))-c(u)\geq q$;
\item $u\in B_i\setminus U$ and $g_u(n_G(u, U)+f(u))\geq q$.
\end{itemize}

\item[Case:~$B_i$ is a forget bag] \hfill

Let~$x$ be the child of~$i$ in~$T$.
Let $B_x\setminus B_i=\{v\}$. If~$i$ is the root, or~$B_i^{\text{P}}$ is a join or an introduce bag, then~${\textsf{DT}}_i(U, f)=1$ if and only if there is a ${\textsf{DT}}_x(U', f')=1$ such that $U'\cap B_i=U$ and $f'|_{\neg \{v\}}=f$.
If~$B_i^{\text{P}}$ is a forget bag, let $B_i\setminus B_i^{\text{P}}=\{u\}$.
Then, we set ${\textsf{DT}}_i(U, f)=1$ if and only if there is a ${\textsf{DT}}_x(U', f')=1$ such that $U'\cap B_i=U$, $f'|_{\neg \{v\}}=f$ and one of the following conditions holds:
\begin{itemize}
\item $u\in U$ and $g_u(n_G[u, U]+f(u))-c(u)\geq q$;
\item $u\in B_i\setminus U$ and $g_u(n_G(u, U)+f(u))\geq q$.
\end{itemize}
\end{description}

After computing the values of all tables, we conclude that the game~$\mathcal{G}$ admits a profile with ESW at least~$q$ if and only if ${\textsf{DT}}_r(\emptyset, \emptyset)=1$ where~$r$ is the root of~$T$.

To see that the algorithm runs in polynomial time, recall first that, as described above, the value of each entry in all tables can be computed in polynomial time. Moreover, there are polynomially many nodes in the tree~$T$, and for each node~$i$, the associated table~${\textsf{DT}}_i$ contains at most $2^{\abs{B_i}} \cdot (\max_{u\in B_i} n_i(u))^{\abs{B_i}}\leq 2^p \cdot n^p$ entries. The  running time follows then from the fact that~$p$ is a constant.

For the whole algorithm, we first find the maximum possible value~$q$ such that~$\mathcal{G}$ admits a profile of ESW at least~$q$, then using standard backtracking technique, a profile with ESW~$q$ can be computed in polynomial time based on the above dynamic programming algorithm.
\end{proof}

\end{document}